\documentclass[conference,a4paper]{IEEEtran}

\usepackage{calc}
\usepackage{graphicx}
\usepackage{color}
\usepackage{url}
\usepackage{amsfonts}
\usepackage{amssymb}
\usepackage{graphicx}
\usepackage{epsfig}
\usepackage{bm} 
\usepackage{bbm} 
\usepackage{epstopdf} 
\usepackage{enumerate}

\newcommand{\beq}[1]{\begin{equation}\label{#1}}
\newcommand{\eeq}{\end{equation}}

\newcommand{\beqn}[1]{\begin{eqnarray}\label{#1}}
\newcommand{\eeqn}{\end{eqnarray}}

\newtheorem{thmbody}{Theorem}
\newenvironment{thm}{
\begin{thmbody}
	}{
	\end{thmbody} 
	}
\newtheorem{dfnbody}{Definition}

\newtheorem{corbody}{Corollary}

\newtheorem{lemmabody}{Lemma}
\newenvironment{lemma}{
\begin{lemmabody}
	}{
	\end{lemmabody} 
	}
\newtheorem{conjbody}{Conjecture}

\newtheorem{propbody}{Proposition}

\newenvironment{proof}{
	{\it Proof:}
	}{
 $\Box$
	}

\usepackage{dsfont}

\usepackage[utf8]{inputenc} 
\usepackage[T1]{fontenc}
\usepackage{url}
\usepackage{ifthen}
\usepackage{cite}
\usepackage{amssymb}
\usepackage[cmex10]{amsmath} 
\usepackage{enumitem}
\usepackage{cancel}
\usepackage{tikz}
\usepackage{tkz-graph}
\usepackage{comment}

\DeclareMathOperator{\tr}{Tr}
\DeclareMathOperator{\Ev}{E}
\newcommand{\R}{\mathbb{R}}
\newcommand{\C}{\mathbb{C}}

\begin{document}
	\title{Frame Moments and Welch Bound with Erasures } 
	
	\author{%
		\IEEEauthorblockN{Marina Haikin}
		\IEEEauthorblockA{EE - Systems Department\\Tel Aviv University\\
			Tel Aviv, Israel\\
			Email: mkokotov@gmail.com}
		\and
		\IEEEauthorblockN{Ram Zamir}
		\IEEEauthorblockA{EE - Systems Department\\Tel Aviv University\\
			Tel Aviv, Israel\\
			Email: zamir@eng.tau.ac.il}
		\and
		\IEEEauthorblockN{Matan Gavish}
		\IEEEauthorblockA{School of Computer Science and Engineering\\
			The Hebrew University\\ 
			Jerusalem, Israel\\
			Email: gavish@cs.huji.ac.il}
	}

	\maketitle
	
	\begin{abstract}
		The Welch Bound is a lower bound on the root mean square cross correlation between $n$ unit-norm vectors $f_1,...,f_n$ in the $m$ dimensional space ($\R ^m$ or $\C ^m$), for $n\geq m$.
		Letting $F = [f_1|...|f_n]$ denote the $m$-by-$n$ frame matrix, the Welch bound can be viewed as a lower bound on the second moment of $F$,
		namely on the trace of the squared Gram matrix $(F'F)^2$.
		We consider an erasure setting, in which a reduced frame, composed of a random subset of Bernoulli selected vectors, is of interest.
		We extend the Welch bound to this setting and present the
		{\em erasure Welch bound} on the expected value of the Gram matrix of the reduced frame.
		Interestingly, this bound generalizes to the $d$-th order moment of $F$.
		We provide simple, explicit formulae for the generalized bound for $d=2,3,4$, which is the sum of the $d$-th moment of Wachter’s classical MANOVA distribution
		and a vanishing term (as $n$ goes to infinity with $\frac{m}{n}$ held constant).
		The bound holds with equality if (and for $d = 4$
		only if) $F$ is an Equiangular Tight Frame (ETF).
		Our results offer a novel perspective on the superiority of ETFs over other
		frames in a variety of applications, including spread spectrum
		communications, compressed sensing and analog coding.
	\end{abstract}
	
	
	\section{Introduction}
	Design of frames or over-complete bases with favorable properties is a thouroughly studied subject in communication, signal processing and harmonic analysis.  In various applications, one is interested in finding over-complete bases
	where the favorable properties hold for a random subset of the frame vectors, rather than for the entire frame.
	
	Here are a few examples. In code-devision multiple access (CDMA), spreading sequences with low cross-correlation are preferred; when only a random subset of the users is active, the quantity of interest is the expected cross-correlation within a random subset of the spreading sequences \cite{rupf1994optimum}.
	In sparse signal reconstruction from undersampled measurements, the ability to reconstruct the signal crucially depends on 
	properties of a subset of the measurement matrix, which corresponds to the non-zero entries of the sparse signal;
	for example, if the extreme eigenvalues of the submatrix are bounded, stable recovery is guaranteed \cite{candes2008restricted}. When the support of the sparse vector is random, one is interested in extreme eigenvalues of a random frame subset \cite{calderbank2010construction}.
	In analog coding, various schemes of interest require frames, for which the first inverse moment of the 
	covariance matrix of a randomly chosen frame subset is as small as possible. This occurs, for example, 
	in the presence of source erasures at the encoder \cite{haikin2016analog},
	in channels with impulses \cite{wolf1983redundancy} or with erasures \cite{ITA17}
	and in multiple description source coding \cite{mashiach2013sampling}.
	
	A famous result by Welch \cite{welch1974lower} provides a universal lower bound on the mean and maximum value 
	of powers of absolute values of inner products (a.k.a cross-correlations) of frame vectors. 
	Frames which achieve the Welch lower bound on maximal absolute cross-correlation
	are known as equiangular tight frame (ETF). 
	
	Motivated by frame design for various applications, in this paper we show that the Welch bound naturally extends to random frame subsets, such that 
	the lower bound is achieved by (and sometimes only by) ETFs. We term this new 
	universal lower bound the
	{\em Erasure Welch Bound} (EWB) and generalize it to higher-order covariances as well.
	
	As a universal, tight lower bound in frame theory, the EWB is essentially a geometric quantity. Surpringly, 
	the EWB itself coincides with a quantity appearing elsewhere in mathematics, namely in random matrix theory.
	Below, we prove that the EWB matches the moments of Wachter's classical limiting MANOVA distribution \cite{wachter1980limiting}. 
	In a recent paper \cite{haikin2017random} we reported overwhelming empirical evidence
	that the covariance matrix of a random frame subset from many well-known ETFs in fact follows the 
	Wachter's classical limiting MANOVA distribution. To the best of our knowledge,
	the results of this paper are the first theoretical confirmations to the empirical predictions of \cite{haikin2017random}, relating ETFs to Wachter's classical limiting MANOVA distribution and random matrix phenomena.
	
	\section{Notation and Setup} \label{definitions}
	We consider a unit-norm frame, being an over-complete basis comprising $n$ elements - unit-norm vectors $f_1,\dots,f_n$.
	Let $F=\{F_{j,i}\}$ denote the $m$-by-$n$ frame matrix whose columns are the frame vectors, $F=\left[f_1|\cdots|f_n\right]$.
	Let us define the vector cross correlation:
	\begin{equation} \label{corr}
	c_{i_1,i_2}\triangleq<f_{i_1},f_{i_2}> = f_{i_1}'f_{i_2}=\sum_{j=1}^{m}F_{j,i_1}^*F_{j,i_2}
	\end{equation}
	where 
	\begin{equation} \label{c_ii}
	c_{i,i} = \|f_i\|^2=1
		\end{equation}
		by the unit norm property.
	The {\em Welch bound} \cite{welch1974lower} lower bounds the root-mean-square (rms) absolute cross correlation:
	\begin{equation} \label{rms WB}
	I^2_{rms}(F) \triangleq \frac{1}{n(n-1)}\sum_{i_1=1}^{n}\sum_{i_2\neq i_1}^{n}|c_{i_1,i_2}|^2\ge \frac{n-m}{(n-1)m},
	\end{equation}
	and it is achieved with equality iff $F$ is a Uniform Tight Frame (UTF), i.e.
	\begin{equation} \label{UTF}
	FF'=\frac{n}{m}I_m.
	\end{equation}
	The Welch bound \cite{welch1974lower} implies a bound on the maximum absolute cross correlation:
	\begin{equation} \label{max WB}
	I^2_{max}(F) \triangleq \max_{1\le i_1< i_2\le n}|c_{i_1,i_2}|^2\ge \frac{n-m}{(n-1)m}.
	\end{equation}
	This stronger lower bound is achieved with equality iff the frame is an Equiangular Tight Frame (ETF), namely, it is UTF \eqref{UTF} and satisfies 
	\begin{equation} \label{ETF2}
	|c_{i_1,i_2}|^2={\rm constant}=\frac{n-m}{(n-1)m} \ \ \forall i_1 \neq i_2.
	\end{equation}
	This unique configuration, which exists only for some dimensions $m$ and number of vectors $n$, achieves a whole family of lower bounds which are derived below.
	
	Our main object of interest is a submatrix composed of a random subset of the frame vectors, or columns of $F$.
	Define the following $m$-by-$n$ matrix
	\begin{equation} \label{X}
	X = FP,
	\end{equation}
	where $P$ is a diagonal matrix with independent Bernoulli($p$) elements on the diagonal.
	In other words, each of the vectors $f_1,...,f_n$ is replaced by a zero vector with probability $1-p$.
	The empirical moment,
		\begin{equation} \label{moment d}
		\frac{1}{n}\tr \left((X'X)^d\right)
		\end{equation}
	is the $d$-th moment of the empirical eigenvalues distribution of $X'X$. 
	We define the expected $d$-th moment of a random subset of $F$ as:
	\begin{equation} \label{moment d}
	m_d \triangleq \frac{1}{n}\Ev\left[ \tr \left((X'X)^d\right)\right]=\frac{1}{n}\Ev\left[ \tr \left((FPF')^d\right)\right]
	\end{equation}
	where we applied $\tr \left((X'X)^d\right)=\tr \left((XX')^d\right)$ and $P^2=P$.

	The first moment ($d=1$) of a frame is constant since
	\begin{equation} \label{m1}
	\begin{split}
	&m_1 = \frac{1}{n}\Ev\left[ \tr \left(X'X\right)\right] = \frac{1}{n}\Ev\left[ \sum_{i=1}^{n}f_i'f_iP_{i,i}\right]\\&=\frac{1}{n}\Ev\left[ \sum_{i=1}^{n}P_{i,i}\right]=\frac{1}{n}\sum_{i=1}^{n}\Ev\left[ P_{i,i}\right]=\frac{1}{n}\sum_{i=1}^{n}p = p
	\end{split}
	\end{equation}
	where the third equality is due to \eqref{c_ii}.
	A useful result for attaining bounds for $d>1$ is the special case of $p=1$, i.e. a bound on the moments of the whole frame without taking subsets.  
	\begin{lemma} \label{lemma1}
		For any unit-norm frame,
		\begin{equation} \label{md bound_p=1}
		\frac{1}{n}\tr \left((FF')^d \right)\ge\left(\frac{n}{m}\right)^{d-1} 
		\end{equation}
		with equality iff $F$ is a UTF.
	\end{lemma}
	\begin{proof}
		The trace of the square matrix $FF'$ is equal to the sum of its eigenvalues $\{\lambda\}_{j=1}^m$. Furthermore, the eigenvalues of $(FF')^d$ are $\{\lambda^d\}_{j=1}^m$. 
		Using Jensen's inequality for a convex function of $\{\lambda\}_{j=1}^m$:
		\begin{equation} \label{Jensen}
		\frac{1}{m}\sum_{j=1}^{m}\lambda^d_{j}\ge \left( \frac{1}{m}\sum_{j=1}^{m}\lambda_{j}\right)^d
		\end{equation}
		with equality iff all eigenvalues are equal, i.e. $FF'\propto I_m$. Hence,
		\begin{equation} \label{Jensen tr}
		\Rightarrow \frac{1}{m}\tr \left((FF')^d \right)\ge \left(\frac{1}{m}\tr(FF')\right)^d 
		\end{equation}
		\begin{equation} \label{trace FF'}
		\frac{1}{m}\tr(FF')=\frac{1}{m}\tr(F'F)=\frac{1}{m}\sum_{i=1}^{n}c_{i,i}=\frac{n}{m}
		\end{equation}
		From \eqref{Jensen tr}, \eqref{trace FF'} and by proper normalization, \eqref{md bound_p=1} follows, with equality iff \eqref{UTF} is satisfied, i.e. $F$ is UTF.
	\end{proof}
	
	Note that for $d=2$, 
	\begin{equation} \label{Lemma d=2}
	\frac{1}{n}\tr \left((FF')^2\right)=\frac{1}{n}\sum_{i_1,i_2=1}^{n}|c_{i_1,i_2}|^2=1+\frac{1}{n}\sum_{i_2\neq i_1}^{n}|c_{i_1,i_2}|^2,
	\nonumber
	\end{equation}
	so \eqref{md bound_p=1} becomes 
	\begin{equation} \label{WB x}
	\frac{1}{n}\sum_{i_1}^{n}\sum_{i_2\neq i_1}^{n}|c_{i_1,i_2}|^2\ge \frac{n}{m}-1\triangleq x
	\end{equation}
	which is the Welch bound \eqref{rms WB}.
	Therefore, a lower bound on $m_d$ in \eqref{moment d} generalizes the Welch bound in two senses. First as a bound on a random subsets of $F$ (where for $p=1$, it reduces to the rms Welch bound). Second, as a bound on higher orders of moments, for $d\ge 2$ \footnote{Our definition is different than that of the Welch bound on the powers of the absolute cross-correlations in \cite{welch1974lower}.}.
	
	\section{Main Result}\label{main}
	To state our main theorem, let us define the $d$-th moment of the MANOVA$(\gamma,p)$ density as, \cite{dubbs2015infinite}
	\begin{equation} \label{moment d MANOVA}
	m^{\rm MANOVA}(\gamma,p,d)\triangleq \min (p,\gamma)\int t^d \,\rho_{p,\gamma}(t)dt
	\end{equation}
	where $\gamma = \frac{m}{n}$ is the aspect ratio of the frame,
	$\min (p,\gamma)$ is due to normalization by full dimension $n$, 
	and
	\begin{equation}	\label{ManovaDensity}
	\begin{split}
	&\rho_{p,\gamma}(t)=\frac{\gamma \sqrt{(t-r_-)(r_+-t)}}{2\pi t(1-\gamma t)\min (p,\gamma)}\cdot 
	I_{(r_-,r_+)}(t) \\&+\left(p+\gamma-1\right)^+/\min (p,\gamma)\cdot \delta(t-\frac{1}{\gamma})
	\end{split}
	\end{equation}
	is Wachter's classical MANOVA desnity \cite{wachter1980limiting}, compactly supported on $[r_-,r_+]$  with
	\begin{equation}
	\label{ManovaDensityExtrimalValues}
	r_\pm=\bigg(\sqrt{\frac{p}{\gamma}(1-\gamma)}\pm\sqrt{1-p}\bigg)^2\,.
	\end{equation} 
	Using $x=\frac{1}{\gamma}-1$ \eqref{WB x}, let:
	\begin{equation} \label{delta EWBd}
	\Delta(\gamma,p,d,n)\triangleq
	\begin{cases}
	0,& d=2,3  \\
	p^2(1-p)^2\frac{x^2}{n-1},& d=4  \\
	\end{cases}.
	\end{equation}
	\begin{thm}[Erasure Welch Bound of order $d$] \label{th1}
		For any  $m$-by-$n$ unit-norm frame and $d=2,3,4$, the $d$-th moment \eqref{moment d} is lower bounded by
		\begin{equation} \label{moment d bound}
		m_d  \ge 
		m^{\rm MANOVA}(\gamma,p,d)+\Delta(\gamma,p,d,n).
		\end{equation}
		with equality for $d=2,3$ iff $F$ is a UTF, and for $d=4$ iff $F$ is an ETF.
	\end{thm}
	~\\
	The Erasure Welch Bound admits a simple closed form.
	We can write the first term in \eqref{moment d bound} for $d=2,3,4$ as
	\begin{align} \label{Manova moments}	
	m^{\rm MANOVA}(\gamma,p,2) &= p+p^2x\\
	\nonumber
	m^{\rm MANOVA}(\gamma,p,3) &= p+p^23x+p^3(x^2-x)\\
	\nonumber
	m^{\rm MANOVA}(\gamma,p,4) &= p+p^26x+p^3(6x^2-4x)
	\nonumber
	\\&+p^4(x^3-3x^2+x)
	\nonumber
	\end{align}
	where x is defined in \eqref{WB x}.
	As for the second term, note that $\Delta(\gamma,p,4,n)\to 0$ as $n\to \infty$. Therefore, the lower bound is asymptotically $m^{\rm MANOVA}(\gamma,p,d)$ for $d=2,3$ and $4$.
	This is in line with the empirical results in \cite{haikin2017random}, where we showed that random subsets of ETFs have MANOVA spectra.
	
	We can see from \eqref{delta EWBd} that $\Delta(\gamma,p=1,d,n)=0$, and from \eqref{Manova moments} that $m^{\rm MANOVA}(\gamma,p=1,d)=(x+1)^{d-1}$. Thus for $p=1$ the bound \eqref{moment d bound} becomes $\left(\frac{n}{m}\right)^{d-1}$ and coincides with Lemma \ref{lemma1}.
	
	For $d=2$, the bound of Theorem \ref{th1} strengthens the Welch bound in the following sense. Let $k=pn$, and in contrast to our general setting of constant aspect ratio $\gamma$, in this discussion $k$ and $m$ are held constant (subset's size and dimension). Let us consider the expected average cross correlation \eqref{rms WB} of a random subset from $F$,
	\begin{equation} \label{rms FP}
	I^2_{rms}(FP) = \frac{1}{(k-1)k}\Ev \left[{\sum_{i_1,i2 \in S} |c_{i_1,i_2}|^2 }\right]
	\end{equation}
	where $S \in {1,...,n}$ is the subset of selected indices
	(the $i$'s for which $P_{i,i}$ =1).
	Note that we normalize by the expected subset size ($k$).
	In view of the definition of the second moment $m_2$ in \eqref{moment d},
	\begin{equation} \label{rms FP}
	I^2_{rms}(FP) = \frac{\frac{m_2}{p} - 1}{k-1}  \ge \frac{\frac{k}{m} - \frac{k}{n}}{k-1},
	\end{equation}
	where the lower bound follows from Theorem \ref{th1}.
	Note that the Welch bound \eqref{rms WB} corresponds to the case $n=k$ and equals to $\frac{\frac{k}{m} - 1}{k-1}$
	,
	while for $n>k$ the lower bound above {\em increases}, and goes to $\frac{k}{(k-1)m}$ in the limit as $n\to \infty$.  Thus, the new bound accounts the penalty in the rms cross correlation due to randomly choosing the vectors from a fixed larger set of vectors \cite{rupf1994optimum}.
	As $n\to \infty$ $(p\to0)$, this bound amounts to choosing the $k$ vectors uniformly over a unit sphere.
	
	Another interesting point of view is provided by
	random matrix theory. The penalty of the erasure Welch bound corresponds to the increase in the MANOVA second moment $m^{MANOVA}(\gamma,p,2)$, as $p$ varies from 1 to zero. And in the limit as $p\to 0$, this becomes the second moment of the Mar\u cenko-Pastur distribution of an i.i.d matrix \cite{tulino2004random}.
	\section{Proof of Theorem \ref{th1}}
		We show by induction that 
		\begin{equation} \label{XX'k}
		\begin{split}
		&\left((XX')^k\right)_{j_1,j_{k+1}} = \sum_{j_2,\dots,j_k}^{m}\sum_{i_1,\dots,i_k}^{n}F_{j_1,i_1}P_{i_1,i_1}F'_{i_1,j_2}\cdot\\&
		F_{j_2,i_2}P_{i_2,i_2}F'_{i_2,j_3}\cdots F_{j_k,i_k}P_{i_k,i_k}F'_{i_k,j_{k+1}} \,\,;\,\, \text{for}\,\, k=1,2,\dots
		\end{split}
		\end{equation}
		The induction basis ($k=1$) trivially holds:
		\begin{equation} \label{XX'1}
		\begin{split}
		&\left(XX'\right)_{j_1,j_2} = \sum_{i_1=1}^{n}F_{j_1,i_1}P_{i_1,i_1}F'_{i_1,j_2}
		\end{split}
		\end{equation}
		For the induction step, let us assume that \eqref{XX'k} holds for $k=d$, and show for $k=d+1$,
		\begin{equation} \label{XX'd+1}
		\begin{split}
		&\left((XX')^{d+1}\right)_{j_1,j_{d+2}} = \sum_{j_{d+1}}^{m}\left((XX')^d\right)_{j_1,j_{d+1}}\left(XX'\right)_{j_{d+1},j_{d+2}}\\&= \sum_{j_2,\dots,j_{d+1}}^{m}\sum_{i_1,\dots,i_{d+1}}^{n}F_{j_1,i_1}P_{i_1,i_1}F'_{i_1,j_2}\cdot\\&
		F_{j_2,i_2}P_{i_2,i_2}F'_{i_2,j_3}\cdots F_{j_{d+1},i_{d+1}}P_{i_{d+1},i_{d+1}}F'_{i_k,j_{d+2}}
		\end{split}
		\end{equation}
		From \eqref{XX'k} it follows that
		\begin{equation} \label{tr XX'k}
		\begin{split}
		&\tr\left((XX')^k\right) = \sum_{j_1}\left((XX')^k\right)_{j_1,j_1}\\&=\sum_{j_1,\dots,j_k}^{m}\sum_{i_1,\dots,i_k}^{n}F_{j_1,i_1}P_{i_1,i_1}F'_{i_1,j_2}\cdot\\&
		F_{j_2,i_2}P_{i_2,i_2}F'_{i_2,j_3}\cdots F_{j_k,i_k}P_{i_k,i_k}F'_{i_k,j_1}\\&=\sum_{j_1,\dots,j_k}^{m}\sum_{i_1,\dots,i_k}^{n}F_{j_1,i_1}F^*_{j_2,i_1}\cdot\\&
		F_{j_2,i_2}F^*_{j_3,i_2}\cdots F_{j_k,i_k}F^*_{j_1,i_k}P_{i_1,i_1}P_{i_2,i_2}\cdots P_{i_k,i_k}.
		\end{split}
		\end{equation}
		For the $d$-th order, we can sum over $j_1,\dots,j_d$ (row indices) and use \eqref{corr}, to obtain the following chain of correlations:
		\begin{equation} \label{tr XX'd}
		\begin{split}
		&\frac{1}{n}\tr\left((XX')^d\right)
		=\frac{1}{n}
		\sum_{i_1,\dots,i_d}^{n}c_{i_1,i_2}c_{i_2,i_3}\cdots c_{i_d,i_1}P_{i_1,i_1}\cdots P_{i_d,i_d}
		\end{split}
		\end{equation}
		In order to take the expectation we break the sum into cases according to possible combinations of distinct or equal indices. When the number of distinct values in $i_1,\dots,i_d$ is $k$, $\Ev \left[P_{i_2,i_2}\cdots P_{i_d,i_d}\right]=p^k$. The sum of $\frac{1}{n}c_{i_1,i_2}c_{i_2,i_3}\cdots c_{i_d,i_1}$ over all such combinations is denoted by $a_{d,k}(F)$.
		Note that for $k=1$, $a_{d,1}(F)=\frac{1}{n}\sum_{i_1=\dots =i_d=i}^{n}c_{i,i}^d=1$. Hence, $m_d$ can be written in the following form:
		\begin{equation} \label{md poly of p}
		\begin{split}
		m_d = p+p^2a_{d,2}(F)+p^3a_{d,3}(F)+\cdots+p^da_{d,d}(F)
		\end{split}
		\end{equation}
		where $a_{d,d}(F)$ is of a special interest, and corresponds to the cycle of correlations of all distinct indices:
		\begin{equation} \label{a_dd(F)}
		\begin{split}
		a_{d,d}(F) = \frac{1}{n} \sum_{i_1\neq i_2\neq i_3\neq ..\neq i_d}^{n}c_{i_1,i_2}c_{i_2,i_3}\cdots c_{i_d,i_1} 
		\end{split}
		\end{equation}
		We now turn to consider each of the special cases $d=2,3,4$.
		\\ \textbf{Second moment}:
		According to \eqref{md poly of p} we have
		\begin{equation} \label{m2}
		\begin{split}
		&m_2 =p+p^2a_{2,2}(F)
		\end{split}
		\end{equation}
		where $a_{2,2}(F)$ correspond to cases with $i_1\neq i_2$
		\begin{equation} \label{a22}
		\begin{split}
		&a_{2,2}(F)=\frac{1}{n}\sum_{i_2\neq i_1}^{n}|c_{i_1,i_2}|^2\ge x
		\end{split}
		\end{equation}
		where the inequality is due to the rms Welch bound \eqref{WB x}, and satisfied with equality iff $F$ is a UTF.
		From \eqref{m2} and \eqref{a22}, 
		\begin{equation} \label{m2 bound}
		\begin{split}
		&m_2 \ge  p+p^2x=m^{\rm MANOVA}(\gamma,p,2).
		\end{split}
		\end{equation}
		\\
		\textbf{Third moment}:
		According to \eqref{md poly of p},
		\begin{equation} \label{m3 calc}
		\begin{split}
		m_3 &= p + p^2a_{3,2}(F) + p^3a_{3,3}(F).
		\end{split}
		\end{equation}
		$a_{3,2}(F)$ includes all combinations of 2 distinct values for $i_1,i_2,i_3$:
		\begin{equation} \label{a32}
		\begin{split}
		a_{3,2}(F)&=\frac{1}{n}3\sum_{i_1=i_2}^{n}\sum_{i_3\neq i_1}^{n}c_{i_1,i_1}c_{i_1,i_3}c_{i_3,i_1}\\&=\frac{1}{n}3\sum_{i_3\neq i_1}^{n}|c_{i_1,i_3}|^2=3a_{2,2}(F),
		\end{split}
		\end{equation}
		where we used $c_{i,i}=1$ and \eqref{a22}.
		Since \eqref{m3 calc} holds for every $p$, we can set $p=1$ and use \eqref{a32}, Lemma \ref{lemma1} for $d=3$ to obtain: 
		\begin{equation} \label{m3_p=1}
		\begin{split}
		1 + 3a_{2,2}(F)+a_{3,3}(F)\ge \left(\frac{n}{m}\right)^2=(x+1)^2
		\end{split}
		\end{equation}
		From \eqref{m3 calc}, \eqref{a32} and \eqref{m3_p=1} 
		\begin{equation} \label{m3 bound}
		\begin{split}
		&m_3\ge p + p^23a_{2,2}(F)+p^3((x+1)^2-1-3a_{2,2}(F)) \\&= (p-p^3)+(p^2-p^3)3a_{2,2}(F)+p^3(x+1)^2
		\end{split}
		\end{equation}
		Since $p\le 1$, we have $p^2-p^3\ge 0$, and we can use \eqref{a22} to get a lower bound on the third moment of a unit norm frame:
		\begin{equation} \label{m3 bound2}
		\begin{split}
		&m_3\ge(p-p^3)+(p^2-p^3)3x+p^3(x+1)^2\\&=p + p^2 3x + p^3 (x^2-x) =m^{\rm MANOVA}(\gamma,p,3) 
		\end{split}
		\end{equation}
		and the condition for equality in both \eqref{a22} and \eqref{m3_p=1} is the frame being a UTF.
		\\ \textbf{Fourth moment}:
		According to \eqref{md poly of p},
		\begin{equation} \label{m4 poly of p}
		\begin{split}
		m_4 &= p + p^2a_{4,2}(F) + p^3a_{4,3}(F) + p^4a_{4,4}(F).
		\end{split}
		\end{equation}
		Denote $h(\{i_l\}_{l=1}^4)=c_{i_1,i_2}c_{i_2,i_3}c_{i_3,i_4}c_{i_4,i_1} $.  
		Considering all partitions of $\{i_l\}_{l=1}^4$ into 2 groups (2 distinct values), we get:
		\[
		a_{4,2}=\underbrace{4\frac{1}{n}\sum_{i_2=i_3=i_4\neq i_1}^{n}h}_{a^{(1)}_{4,2}}
		+\underbrace{2\frac{1}{n}\sum_{i_1=i_2\neq i_3=i_4}^{n}h}_{a^{(2)}_{4,2}}
		+\underbrace{\frac{1}{n}\sum_{i_1=i_3\neq i_2=i_4}^{n}h}_{a^{(3)}_{4,2}}
		\]
		where $a^{(1)}_{4,2}$ corresponds to partitions consisting of 3 identical indices a 1 different - $i_1$ or $i_2$ or $i_3$ or $i_4$, $a^{(2)}_{4,2}$ corresponds to partitions consisting of 2 different, non-crossing, pairs of indices - $i_1=i_2,i_3=i_4$ or $i_2=i_3,i_4=i_1$, $a^{(3)}_{4,2}$ corresponds to a partition consisting of 2 different, crossing, pairs of indices - $i_1=i_3,i_2=i_4$.
		We derive now the three components:
		\begin{align} \label{a42 components}	
		&a^{(1)}_{4,2}=4\frac{1}{n}\sum_{i_2\neq i_1}^{n}|c_{i_1,i_2}|^2=4a_{2,2}(F)\\
		&a^{(2)}_{4,2}=2\frac{1}{n}\sum_{i_3\neq i_1}^{n}|c_{i_1,i_3}|^2=2a_{2,2}(F)\\
		\label{a42_3}
		&a^{(3)}_{4,2}=\frac{1}{n}\sum_{i_2\neq i_1}^{n}|c_{i_1,i_2}|^4 
		\end{align}
		We lower bound $a^{(3)}_{4,2}$. By Jensen's inequality:
		\begin{equation} \label{a42_3_bound}
		\begin{split}
		\frac{1}{n(n-1)}\sum_{i_2\neq i_1}^{n}|c_{i_1,i_2}|^4\ge \left(\frac{1}{n(n-1)}\sum_{i_2\neq i_1}^{n}|c_{i_1,i_2}|^2\right)^2
		\end{split}
		\end{equation}
		which is achieved with equality if all absolute correlations are constant, i.e. $F$ is ETF. Hence, from \eqref{a42_3}, \eqref{a42_3_bound}:
		\begin{equation} \label{a42_3 bound2}
		a^{(3)}_{4,2}\ge \frac{1}{n-1}\left(\frac{1}{n}\sum_{i_2\neq i_1}^{n}|c_{i_1,i_2}|^2\right)^2\ge \frac{x^2}{n-1} 
		\end{equation}
		where the second inequality follows from Welch bound \eqref{WB x}.
		Considering all partitions of $\{i_l\}_{l=1}^4$ into 3 groups, i.e. 3 distinct values, we get:
		\[
		a_{4,3}=\underbrace{4\frac{1}{n}\sum_{i_1=i_2\neq i_3\neq i_4}^{n}h}_{a^{(1)}_{4,3}}
		+\underbrace{2\frac{1}{n}\sum_{i_1=i_3\neq i_2\neq i_4}^{n}h}_{a^{(2)}_{4,3}}
		\]
		where $a^{(1)}_{4,3}$ corresponds to partitions consisting of 1 pair of identical indices and 2 different values- $i_1=i_2$ or $i_2=i_3$ or $i_3=i_4$ or $i_4=i_1$, $a^{(2)}_{4,3}$ corresponds to partitions consisting of 1 pair of identical indices and 2 different values- $i_1=i_3$ or $i_2=i_4$.
		We derive now these two components:
		\begin{align} \label{a43 components}	
		&a^{(1)}_{4,3}=4\frac{1}{n}\sum_{i_2\neq i_3\neq i_4}^{n}c_{i_2,i_3}c_{i_3,i_4}c_{i_4,i_1}=4a_{3,3}(F)\\
		&a^{(2)}_{4,3}=2\frac{1}{n}\sum_{i_1\neq i_2\neq i_4}^{n}|c_{i_1,i_2}|^2|c_{i_1,i_4}|^2
		\end{align}
		Denote $C_{i_1}$ as the sum over all absolute correlations between $i_1$ and other frame vectors.
		\begin{equation} \label{C_i}
		\begin{split}
		C_{i_1} = \sum_{i_2\neq i_1}^{n}|c_{i_1,i_2}|^2
		\end{split}
		\end{equation}
		We derive a lower bound on the sum $a^{(3)}_{4,2}+\frac{1}{2}a^{(2)}_{4,3}$
		\begin{equation} \label{b2_bound}
		\begin{split}
		&\frac{1}{2}a^{(2)}_{4,3}
		=\frac{1}{n}\sum_{i_1}^{n}\sum_{i_2\neq i_1}^{n}|c_{i_1,i_2}|^2\sum_{i_4\neq i_2,i_1}^{n}|c_{i_1,i_4}|^2\\&=\frac{1}{n}\sum_{i_1}^{n}\sum_{i_2\neq i_1}^{n}|c_{i_1,i_2}|^2\left[C_{i_1}-|c_{i_1,i_2}|^2\right]\\&=\frac{1}{n}\sum_{i_1}^{n}C_{i_1}\sum_{i_2\neq i_1}^{n}|c_{i_1,i_2}|^2-\frac{1}{n}\sum_{i_1}^{n}\sum_{i_2\neq i_1}^{n}|c_{i_1,i_2}|^4\\&=\frac{1}{n}\sum_{i_1}^{n}C_{i_1}^2-a^{(3)}_{4,2} \,\,\,\,\,\,\,\,\,   \Rightarrow
		\end{split}
		\end{equation}
		\begin{equation} \label{b1+b2}
		\begin{split}
		a^{(3)}_{4,2}+\frac{1}{2}a^{(2)}_{4,3} &= \frac{1}{n}\sum_{i_1}^{n}C_{i_1}^2 \ge \left(\frac{1}{n}\sum_{i_1}^{n}C_{i_1}\right)^2
		\ge x^2
		\end{split}
		\end{equation}
		where the first inequality is again due to Jensen and is achieved with equality if $C_i$ are equal for all $i$, and the second inequality is the Welch bound \eqref{WB x}.
		Combining all terms we have
		\begin{equation} \label{m4 total}
		\begin{split}
		&m_4=p+p^2(6a_{2,2}+a^{(3)}_{4,2})+p^3(4a_{3,3}+a^{(2)}_{4,3})+p^4a_{4,4}
		\end{split}
		\end{equation}
		Now we repeat the procedure from the bound on $m_3$ with sequential substitution of all bounds and gathering of similar terms.
		We set $p=1$ in \eqref{m4 total} and use Lemma \ref{lemma1}:
		\begin{equation} \label{a44 bound}
		\begin{split}
		&a_{4,4}\ge (x+1)^3-1-6a_{2,2}-4a_{3,3}-a^{(3)}_{4,2}-a^{(2)}_{4,3}
		\end{split}
		\end{equation}
		Substituting \eqref{a44 bound} into \eqref{m4 total} we get:
		\begin{equation} \label{m4 bound2}
		\begin{split}
		&m_4\ge p-p^4+p^4(x+1)^3+(p^2-p^4)6a_{2,2}\\&+(p^2-p^4)a^{(3)}_{4,2}+(p^3-p^4)a^{(2)}_{4,3}+(p^3-p^4)4a_{3,3}.
		\nonumber
		\end{split}
		\end{equation}
		As $(p^3-p^4)\ge 0$, we can substitute \eqref{m3_p=1} and get:
		\begin{equation} \label{m4 bound3}
		\begin{split}
		&m_4\ge p-p^4+p^4(x+1)^3+(p^3-p^4)\left(4(x+1)^2-4\right)\\&+p^2(1-p)6a_{2,2}+(p^2-p^4)a^{(3)}_{4,2}+(p^3-p^4)a^{(2)}_{4,3}
		\end{split}
		\end{equation}
		and now we use the bound on $a_{2,2}$ \eqref{a22}.
		The last two terms can be reordered to become a function of $a^{(3)}_{4,2}$ and $a^{(3)}_{4,2}+\frac{1}{2}a^{(2)}_{4,3}$ for which we have bounds
		\begin{equation} \label{m4 bound4}
		\begin{split}
		&(p^2-p^4)a^{(3)}_{4,2}+(p^3-p^4)a^{(2)}_{4,3} \\&=(p^3-p^4)2(a^{(3)}_{4,2}+\frac{1}{2}a^{(2)}_{4,3})+p^2(1-p)^2a^{(3)}_{4,2}
		\end{split}
		\end{equation}
		So now we can apply \eqref{a42_3 bound2} and \eqref{b1+b2}
		\begin{equation} \label{m4 bound5}
		\begin{split}
		&m_4\ge p+p^2(6x+\frac{1}{n-1}x^2)+p^3(6x^2-4x-2\frac{1}{n-1}x^2)\\&+p^4(x^3-3x^2+x+\frac{1}{n-1}x^2)\\&=m^{\rm MANOVA}(\gamma,p,4)+p^2(1-p)^2\frac{x^2}{n-1}
		\end{split}
		\end{equation}
		with equality iff $F$ is ETF. 
		Note that the asymptotic lower bound $\lim\limits_{n\to \infty}m_4\ge m^{\rm MANOVA}(\gamma,p,4)$ holds with equality under the weaker condition that $F$ is a UTF and $C_{i_1} = \sum_{i_2\neq i_1}^{n}|c_{i_1,i_2}|^2$ is equal for all $i$ and $\frac{1}{n}\sum_{i_2\neq i_1}^{n}|c_{i_1,i_2}|^4\to 0$ as $n\to \infty$, i.e. ETF is sufficient but not necessary. 
 $\Box$
	\section{Discussion and Future Work}\label{discussion}
	We are currently working to extend our results to higher order moments $d$. For example, we already found the asymptotic form for the moments of orders $d=5,6$ of subsets of ETF, and verified that they agree with that of MANOVA. Furthermore, we developed a recursive procedure which allows to continue to higher order moments. A complete computation of all the moments will provide formal validation for some of the empirical results reported in \cite{haikin2017random}, and specifically, that the singular values of random subsets of an ETF asymptotically follow Wachter's MANOVA distribution.
	The performance of analog coding \cite{haikin2016analog,ITA17} relies on yet another figure of merit of frame subsets, namely the harmonic-to-arithmetic means ratio of the singular values of the subframe covariance matrix. In our standing notation, this quantity is equivalent to the first inverse moment $d=-1$.  Extension of the Erasure Welch Bounds to higher order moments and $d=-1$ would establish that an ETF is the most robust
	frame under inversion of subsets.
	A more complete description of these extensions will 
	appear elsewhere.
	\section*{Acknowledgment}
	We would like to thank Ofer Zeitouni for proposing the moment method for analyzing subsets of ETF. We also thank Benny Zaidel for a helpful discussion. This work has been partially supported by the Israeli Science Foundation grants no. 1523/16, 676/15.

	\bibliographystyle{IEEEtran}
	\bibliography{bibliofile}

	
	
	

\end{document}